\documentclass[twocolumn,prl,showpacs,amsmath,amssymb,aps]{revtex4-1}

\usepackage{bbm}
\usepackage{graphicx}
\usepackage{amssymb,amsthm,amsfonts,amstext}
\usepackage{hyperref}

\newtheorem{proposition}{Proposition}

\newtheorem{lemma}{Lemma}

\def\be{\begin{equation}}
\def\ee{\end{equation}}
\def\bea{\begin{eqnarray}}
\def\eea{\end{eqnarray}}
\def\H{{\cal H}}
\def\proj#1{\ket{#1}\!\bra{#1}}

\let\originalleft\left
  \let\originalright\right
\renewcommand{\left}{\mathopen{}\mathclose\bgroup\originalleft}
  \renewcommand{\right}{\aftergroup\egroup\originalright}

\newcommand{\id}{\mathbbm{1}}                %Identity
\newcommand{\tr}[1]{\operatorname{Tr}\left[ {#1} \right]} %Trace
\newcommand{\trx}[2]{\operatorname{Tr}_{#2}\left[ {#1} \right]} %Trace arbitrary
\newcommand{\ket}[1]{\left|#1 \right \rangle \vphantom{\left( #1 \right)^A}} %ket
\newcommand{\bra}[1]{\left\langle #1 \right | \vphantom{\left(#1\right)^A}} %bra

\allowdisplaybreaks[4]
\begin{document}
\title{Non-thermal quantum channels as a thermodynamical resource}
\author{Miguel Navascu\'{e}s$^1$ and Luis Pedro Garc\'{i}a-Pintos$^2$}
\affiliation{$^1$Department of Physics, Bilkent University, Ankara 06800, Turkey\\
$^2$School of Mathematics, University of Bristol, University Walk, Bristol BS8 1TW, U.K.}
\begin{abstract}
Quantum thermodynamics can be understood as a resource theory, whereby thermal states are free and the only allowed operations are unitary transformations commuting with the total Hamiltonian of the system. Previous literature on the subject has just focused on transformations between different state resources, overlooking the fact that quantum operations which do not commute with the total energy also constitute a potentially valuable resource. In this Letter, given a number of non-thermal quantum channels, we study the problem of how to integrate them in a thermal engine so as to distill a maximum amount of work. We find that, in the limit of asymptotically many uses of each channel, the distillable work is an additive function of the considered channels, computable for both finite dimensional quantum operations and bosonic channels. We apply our results to bound the amount of distillable work due to the natural non-thermal processes postulated in the Ghirardi-Rimini-Weber (GRW) collapse model. We find that, although GRW theory predicts the possibility to extract work from the vacuum at no cost, the power which a \emph{collapse engine} could in principle generate is extremely low.

\end{abstract}

\maketitle

The field of quantum thermodynamics has seen a surge in interest in the past years, with increasing attention towards testing the validity of the rules of classical thermodynamics in the quantum regime. A major topic within thermodynamics is that of extracting work out of a given system and the optimal way to perform this. One way this has been approached in the quantum case was by considering it from the perspective of a resource theory.

The idea of a resource theory of thermodynamics is to assume one has unlimited access to thermal baths (i.e. Gibbs states of a fixed temperature $T$), the freedom to apply any energy conserving unitary on system plus the bath (any unitary $V$ that commutes with the total Hamiltonian), and the possibility of discarding part of the system or bath (i.e. apply partial traces). These rules are imported from classical thermodynamics, where one assumes access to infinite baths of constant temperature and any evolution where energy is conserved.

As a matter of fact, resource theories have been very useful in different topics within quantum information theory~\cite{Horodecki2003,Speckens2008,Horodecki2009,Muller13}. The idea is similar to above: considering free access to certain operations and/or states, any state and/or operation that is not in the above set can in principle be used as a resource.

This work complements previous research in quantum thermodynamics by accommodating the possibility of considering non-thermal maps, or channels, as a resource.

Physical operations are represented by \emph{quantum channels}, i.e., completely positive trace preserving maps $\Omega: B(\mathcal{H}) \rightarrow B(\mathcal{H}')$ acting on a state space $B(\mathcal{H})$. For simplicity, we will assume that the input and output spaces (and, as we will see later, Hamiltonians) of each channel are the same, although the results can be easily generalized. 
%to the case where $\mathcal{H}\not=\mathcal{H}'$.

Unitary evolution is a particular instance of a quantum channel, determined by the evolution operator. However, quantum channels allow to express more general evolutions. For instance, any unitary interaction of a system with an ancilla (or environment) generates a quantum channel, given by
\be
\Omega (\rho) = \trx{V \rho\otimes \sigma_A V^\dag}{A},
\ee
where $\sigma_A$ is the state of the ancilla, and $V$ is some unitary operator %(for instance, $V$ could be the evolution operator under some Hamiltonian). 
In fact, it can be shown that any channel can be generated via the above procedure~\cite{nielsenchuang2010}.

If, in the above expression, the ancilla is in a Gibbs state of temperature $T$ and the unitary $V$ commutes with the total Hamiltonian $H_T=H_S\otimes \id_A+\id_S\otimes H_A$ of the target-ancilla system, the resulting map is called a \emph{thermal channel}. These will be the free operations in our theory, while any non-thermal map $\Omega$ will be considered as a resource. 

In this scenario, we define quantum work $W$ as the process of exciting a two-level system with Hamiltonian $H=W\proj{1}$ from its ground state $\ket{0}$ to the excited state $\ket{1}$~\cite{Horodeckisingleshot}. Different authors have explored how much work one can extract from a non-thermal quantum state~\cite{Horodeckisingleshot,Brandao11,Takara2010,Esposito2011,Aberg13,Brandao2013,Popescu2014}. When we regard the \emph{maximum average work} as a figure of merit, a quantum generalization of the classical free energy naturally emerges~\cite{Horodeckisingleshot,Brandao11,Popescu2014}:
\be
F(\rho) = U(\rho) - K_BT S(\rho).
\ee
Here $\rho$ is the state of the system from which we wish to extract work; $U(\rho) = \tr{\rho H}$, its average energy%with respect to the Hamiltonian $H$
; and $S(\rho) = - \tr{\rho \log{\rho}}$, its von Neumann entropy. $K_B$ is Boltzmann's constant. The maximum amount of work one can extract (on average) from the state $\rho$ can then be shown to be $F(\rho)-F(\tau_{th})$, where $\tau_{th}$ represents the thermal state 
%of the system 
at temperature $T$.

In this Letter we want to address the following related problem: suppose we want to build a thermal engine, where we are allowed to integrate a number of non-thermal gates $\{\Omega_i\}_{i=1}^N$, each of which is assumed to act on a system with Hamiltonian $H_i$. More specifically, our machine can make free use of any amount of thermal states and operations, and we can invoke one use of each of the channels $\{\Omega_i\}_{i=1}^N$, in any order we want at any step. We are also allowed to use \emph{catalysts}, i.e., we can use any number of non-thermal states, as long as we return them in the end. Under these conditions, what is the maximum amount of work that our device can extract?

There are two ways to approach this problem:

\begin{enumerate}
\item
We can restrict to thermodynamical processes which distill work deterministically, i.e., always the same amount. The corresponding \emph{deterministic extractable work} can then be shown to behave very badly: not only is it not additive, but it can be super-activated. That is, there exist channels $\Omega$ such that no work can be distilled from a single use,
%with a single use of $\Omega$, 
but two uses of the channel can be combined to produce a non-zero amount of deterministic work (see the Supplemental Information).

\item
Alternatively, we can consider thermodynamical processes which generate a given amount of work with high probability. Here the figure of merit would be the maximum amount of work that can be distilled asymptotically (on average) when we have access to $n$ uses of each channel. 
\end{enumerate}

We will follow the second approach: in the next pages we will show that the \emph{asymptotically extractable work} is upper bounded by $\sum_{i=1}^N W(\Omega_i,H_i)$, where

\be
W(\Omega, H)\equiv \max_{\rho} \Delta F(\rho, \Omega),
\label{basic}
\ee

\noindent
with $\Delta F(\rho, \Omega)$ denoting the free energy difference between the states $\Omega(\rho)$ and $\rho$, i.e.
\be
\Delta F(\rho, \Omega)\equiv \tr{ (\Omega(\rho)-\rho )H} - K_BT\left[ S( \Omega(\rho))  - S(\rho)\right].
\ee
The quantity $W(\Omega, H)$ will be called the \emph{distillable work} of channel $\Omega$. From the inequality $\Delta F(\tau_{th}, \Omega)\geq 0$, it follows that $W(\Omega, H)\geq 0$ for any $\Omega$.

\noindent  The bound $\sum_{i=1}^N W(\Omega_i,H_i)$ can be achieved asymptotically via a simple protocol where we prepare suitable initial states $\sigma_{cat}$ (the catalysts) maximizing eq.~(\ref{basic}) for each channel, and then let each channel act over its corresponding maximizer. The result of this protocol will be a state with free energy $F(\sigma_{cat})+\sum_{i=1}^NW(\Omega_i,H_i)$. Given access to $n$ uses of each channel, we can thus prepare $n$ copies of the latter state, whose free energy can be converted to work via thermal operations using the protocol depicted in~\cite{Brandao11}. Following \cite{Brandao11}, part of this work (roughly $nF(\sigma_{cat})$) can then be used to regenerate the catalysts up to a small error~\footnote{Crucially, at the end of the regeneration step the free energy of the reconstructed catalysts also tends to its initial value.}. The average work extracted with this procedure (namely, the total work divided by $n$) is thus given by $\sum_{i=1}^N W(\Omega_i,H_i)$. 

Note, though, that, unless the catalysts are diagonal in the energy basis, an extra amount of coherence, sub-linear in $n$, may be needed to rebuild them (see Appendix E of \cite{Brandao11}). More specifically, for each energy transition $E_s\to E_t$ in the Hamiltonian $H_i$, the protocol proposed in \cite{Brandao11} requires a system with Hamiltonian $H^{s,t}_i= \sum_{k=0}^{O(m)} (E_s-E_t) k \proj{k}$ in state $\frac{1}{\sqrt{m}}\sum_{k=0}^m\ket{k}$, with $m$ sublinear in the number $n$ of uses of each channel. Like the catalyst states, at the end of the protocol such `coherent states' will be approximately rebuilt with vanishing error.

% $\ket{\psi^{s,t}_i}=\frac{1}{\sqrt{m}}\sum_{k=0}^m\ket{m}$ with $m$ sublinear in $n$, and Hamiltonian $H^{s,t}_i=\sum_{k=0}^m (E_i-E_j)\proj{m}$. .

%The quantity $W(\Omega, H)$ will be called the \emph{distillable work} of channel $\Omega$. \textbf{XXXXX}
%XXXXXXXXXXXX
%The quantity $W(\Omega, H)$ will be called the \emph{distillable work} of channel $\Omega$. \textbf{Note that the distillable work is always non-negative. Indeed, set $\rho=\tau_{th}$ in equation (\ref{basic}), where $\tau_{th}$ denotes the Gibbs state at temperature $T$. Then we have that $W(\Omega, H)\geq F(\Omega(\tau_{th}))-F(\tau_{th})\geq 0$.}
%XXXXXXXX

In order to prove the above result, and some later ones, the next lemma will be invoked extensively:

\begin{lemma}
\label{basic_form}
Let $\sigma^{(N)}$ be an $N$-partite quantum state, and let $\{\Omega_i\}_{i=1}^N$ be a collection of $N$ single-site quantum channels. Defining $\Omega_{1...N}\equiv\bigotimes_{i=1}^N\Omega_i$, we have that

\be
\sum_{i=1}^NS(\sigma_i)-S(\Omega_i(\sigma_i))\geq S(\sigma^{(N)})-S(\Omega_{1...N}(\sigma^{(N)})).
\label{indep_max}
\ee

\end{lemma}

\noindent The proof is a straightforward application of the contractivity of the relative entropy~\cite{Vedral2002}.

An almost immediate consequence of Lemma \ref{basic_form} is that $W(\Omega, H)$, as defined by eq.~(\ref{basic}), has the remarkable property of being additive. That is, if the bipartite system $12$ is described by the Hamiltonian $H_{12}=H_1\otimes \id_2+\id_1\otimes H_2$, and the channels $\Omega_1$ and $\Omega_2$ act on the respective Hilbert spaces $\H_1,\H_2$, then, $W(\Omega_1\otimes \Omega_2,H_{12})=W(\Omega_1,H_1)+W(\Omega_2,H_2)$.

Indeed, let $\Omega_{12}\equiv\Omega_1\otimes\Omega_2$ act on the bipartite state $\rho_{12}$. By choosing $\rho_{12}=\rho_1\otimes\rho_2$ in (\ref{basic}) we trivially have that $W(\Omega_1\otimes \Omega_2,H_{12})\geq W(\Omega_1,H_1)+W(\Omega_2,H_2)$, since maximizing over states in $12$ is more general than maximizing over $1$ and $2$ independently. Let us then focus on the opposite inequality. By Lemma \ref{basic_form}, we have
\be
\sum_{i=1,2}S(\rho_{i})-S(\Omega_i(\rho_{i}))\geq S(\rho_{12})-S(\Omega_{12}(\rho_{12})).
\ee

\noindent Substituting into~(\ref{basic}) gives
\be
\sum_{i=1,2}\Delta F(\rho_i,\Omega_i) \geq \Delta F(\rho_{12},\Omega_{12}) \quad \ \forall \  \rho_{12},
\ee
%for all $\rho_{12}$. 
It follows that $\sum_{i=1,2}W(\Omega_{i},H_i)\geq W(\Omega_{12},H_{12})$.

We are now ready to prove 
%our first result, namely, 
that $W(\Omega, H)$ quantifies the maximum (average) amount of work one can extract from channel $\Omega$.

\begin{proposition}
\label{max_work}
Let $\{\Omega_i\}_{i=1}^N$ be a set of quantum channels, defined over different quantum systems with Hamiltonians $\{H_i\}_{i=1}^N$. Suppose that we integrate $n$ uses of all such channels in a thermal engine ${\cal T}_n$ that produces a net amount of work $W_n$ with probability $1-\epsilon_n$. Let us further assume that  the probability of failure vanishes in the limit of large $n$, i.e., $\lim_{n\to\infty}\epsilon_n= 0$. Under these conditions, the average asymptotic work $\bar{W}\equiv\limsup\limits_{n\to\infty}\frac{W_n}{n}$ satisfies
\be
\bar{W}\leq \sum_{i=1}^N W(\Omega_i,H_i).
\label{boundido}
\ee

As indicated above, this bound is achievable with the use of catalysts and a sublinear amount of quantum coherence.
\end{proposition}

\begin{proof}

In any protocol for work extraction, the initial state of the system will be given by the catalysts $\sigma_{cat}$, a number of thermal states $\tau_{th}$, and the work system in state $\ket{0}_{w}$. The initial state of the system is hence $\rho_0\equiv \sigma_{cat}\otimes\tau_{th}\otimes \proj{0}_{w}$, with free energy $F(\sigma_{cat})+F(\tau_{th})$.

Suppose that now we apply a sequence of energy-conserving unitaries. At time $t$, the state of the overall system is $\rho_t$, and we apply the channel $\Omega_{s(t)}$ over part of the whole system, possibly followed by some other thermal operation. Let us analyze how the free energy of $\rho_t$ can increase in the above step. Calling $H_T$ the Hamiltonian of the whole system, from the definition of $W(\Omega,H)$ and the additivity of the distillable work we have that:
\be
\Delta F(\Omega_{s(t)}\otimes\id,\rho_t)\leq W(\Omega_{s(t)}\otimes \id,H_T)=W(\Omega_{s(t)},H_{s(t)}).
\ee

Now, any intermediate energy-conserving unitary in-between the use of any two of the channels $\{\Omega_i\}_{i=1}^N$ will keep the free energy of the overall system constant. Calling $\bar{\rho}$ the state of the system at the end of the protocol, we hence have that

\be
F(\bar{\rho})\leq n\sum_{i=1}^NW(\Omega_i,H_i)+F(\sigma_{cat})+F(\tau_{th}).
\ee

\noindent From the subadditivity of the von Neumann entropy, it follows that $F(\bar{\rho})\geq F(\bar{\rho}_{cat})+F(\bar{\rho}_{th})+F(\bar{\rho}_{w})$, where $\bar{\sigma}_{cat},\bar{\rho}_{th},\bar{\rho}_{w}$ are, respectively, the reduced density matrices of the catalyst, thermal and work systems. 

At the end of the protocol, the catalyst must be regenerated, i.e., $\bar{\sigma}_{cat}=\sigma_{cat}$. Also, $F(\bar{\rho}_{th})\geq F(\tau_{th})$. It follows that the free energy of the work system is bounded by $n\sum_{i=1}^NW(\Omega_i,H_i)$.

This system is expected to end up in state $\ket{1}$ with probability $1-\epsilon_n$, i.e., $\bar{\rho}_w=(1-\epsilon_n)\proj{1}+\epsilon_n \bar{\sigma}$. It follows that $F(\bar{\rho}_{w})\geq (1-\epsilon_n)W_n-K_BTh(\epsilon_n)$, with $h(p)=-p\ln(p)-(1-p)\ln(1-p)$. In the asymptotic limit, with $n\to \infty$, $\epsilon_n\to 0$, the average asymptotic work $\limsup\limits_{n\to\infty}\frac{W_n}{n}$ is hence bounded by $\sum_{i=1}^NW(\Omega_i,H_i)$.

Note that this bound also holds if the catalysts are recovered up to an error, as long as $F(\sigma_{cat})-F(\bar{\sigma}_{cat})\leq o(n)$.

\end{proof}

This result allows to quantify the work extraction capabilities of different channels. One can check, for instance, that no work can be distilled from a dephasing channel. Meanwhile, for a two-level system with Hamiltonian $H = E\ket{1}\bra{1},\ E>0$, the channel that takes any state to the excited state $\ket{1}$ provides the highest distillable work.

\vspace{10pt}
\noindent \emph{Gaussian channels}
\vspace{10pt}

If our target system is infinite dimensional, in principle there may exist quantum states possessing an infinite amount of energy. If we regard such states as unphysical, we should replace the maximization in eq.~(\ref{basic}) by an optimization over all states of finite energy. The resulting quantity will hence bound the maximum amount of work generated in physically conceivable quantum engines, where the overall state of the system always has a finite amount of energy.
%(that is, over physically conceivable thermal engines).

In infinite dimensional systems Gaussian quantum channels have a special relevance: they are easy to implement in the lab, and are extensively used to model particle interactions with a macroscopic environment. They are defined as channels which, when composed with the identity map, transform Gaussian states into Gaussian states, the latter being those states with a Gaussian Wigner function~\cite{gaussianos}. An $m$-mode Gaussian state is completely defined via its displacement vector $d_i=\langle R_i\rangle$ and covariance matrix $\gamma_{ij}=\langle\{R_i - d_i \id,R_j - d_j \id\}_{+}\rangle$, where $(R_1,R_2,...,R_{2m})\equiv (Q_1,P_1,...,Q_m,P_m)$ are the optical quadratures. The action of a Gaussian channel is fully specified by its action over the displacement vector and covariance matrix, given by:

\be
d\to Xd+z, \gamma\to X\gamma X^T+Y,
\label{gauss_chan}
\ee

\noindent where $Y+i\sigma-iX^T\sigma X\geq 0$. Here $\sigma$ denotes the symplectic form $\sigma=\oplus_{i=1}^m\left(\begin{array}{cc}0&1\\-1&0\end{array}\right)$. If the Hamiltonian of the system under study happens to be a quadratic function of the optical quadratures, i.e., $H=\vec{R}^TG\vec{R}+\vec{h}\cdot \vec{R}$, for some real symmetric matrix $G$ and real vector $\vec{h}$, then the average energy of a state with displacement vector $\vec{d}$ and covariance matrix $\gamma$ is given by $E=\frac{1}{2}\tr{G \gamma}+\vec{d}^TG\vec{d}+\vec{h}\cdot \vec{d}$. States with finite energy hence correspond to states with finite first and second moments.

When the quadratic Hamiltonian has no zero energy modes (that is, when $G>0$), Proposition \ref{max_work} allows to easily classify generic Gaussian channels according to their capacity to generate an infinite amount of work. Indeed, for $X^TGX-G\not\leq 0$, the channel's distillable work is unbounded: this can be seen by inputting a sequence of Gaussian states with constant covariance matrix but increasing displacement vector parallel to any positive eigenvector of $X^TGX-G$. Conversely, as we show in the Supplemental Information, for channels satisfying $X^TGX-G < 0$ only a finite amount of work can be distilled.

For such channels there is still the dilemma of how much work can be extracted. The next Proposition greatly simplifies this problem by showing that, for Gaussian channels $\Omega$, the maximization in (\ref{basic}) can be restricted to Gaussian states:

%XXXXXXXXXXXXX
%When the quadratic Hamiltonian has no zero energy modes (that is, when $G>0$), Proposition \ref{max_work} allows to easily classify Gaussian channels according to their capacity to generate an infinite amount of work. Indeed, consider eq.~(\ref{gauss_chan}). For generic channels, there are two possibilities:
%
%\begin{enumerate}
%\item
%$X^TGX-G\not\leq 0$, in which case the channel allows to distill an infinite amount of work. This can be seen by inputting a sequence of Gaussian states with constant covariance matrix but increasing displacement vector parallel to any positive eigenvector of $X^TGX-G$.
%\item
%$X^TGX-G<0$, in which case the channel will just allow to distill a finite amount of work. This is a consequence of the fact that, for a given energy $E_0$ of the initial state, the distillable work of the channel is bounded by an expression of the form $O(\log(E_0))-O(E_0)$, see the Supplemental Information for details.
%
%\end{enumerate}
%
%For Gaussian channels with finite distillable work, there is still the dilemma of how much work we can exactly extract from them. The next Proposition greatly simplifies this problem by showing that, for Gaussian channels $\Omega$, the maximization carried out in (\ref{basic}) can be restricted to Gaussian states:
% XXXXXXXXXX

\begin{proposition}
\label{Gaussian}
Consider a continuous variable quantum system of $m$ modes, with quadratic Hamiltonian $H$, let $\Omega$ be a Gaussian channel mapping $m$ modes to $m$ modes, and denote by ${\cal G}$ the set of all $m$-mode Gaussian states. Then,
\be
W(\Omega,H)= \max_{\rho\in {\cal G}} \Delta F(\rho,\Omega).
\ee
\end{proposition}

\noindent The proposition can be proven by combining Lemma \ref{basic_form} with the `gaussification' protocol described in \cite{gaussification}, see the Supplemental Information.

Since $W(\Omega,H)$ just involves an optimization over a finite set of parameters subject to positive semidefinite constraints, (in principle) it can be computed exactly for any Gaussian channel $\Omega$.

\vspace{10pt}
\noindent \emph{One application: collapse engines}
\vspace{10pt}

In order to address the measurement problem~\cite{GRW}, and, independently, the decoherence effects that a quantum theory of gravity could impose on the wave-function~\cite{penrose,diosi}, different authors have proposed that \emph{closed} quantum systems should evolve according to the Lindblad equation

\be
\label{GRWapprox}
\frac{d}{dt} \rho_t = -\frac{i}{\hbar} [H,\rho_t] - \frac{\Lambda}{4} [X,[X,\rho_t]],
\ee
where $X$ is the position operator for the particle considered. 

The effect of the non-unitary term is a suppression of coherences in the position basis, effectively destroying quantum superpositions. The value of the constant $\Lambda$, which can be interpreted as the rate at which this localization process occurs, depends on the particular theory invoked to justify eq. (\ref{GRWapprox}). In the Ghirardi-Rhimini-Weber (GRW) theory~\cite{GRW}, the localization process is postulated to solve the measurement problem in quantum mechanics. To achieve this goal and avoid contradictions with past experimental results, $\Lambda$ must be roughly between $10^{-2} s^{-1} m^{-2} \ \text{and} \ 10^{6} s^{-1} m^{-2}$, according to latest estimations \cite{Adler07}.

Note that the above evolution is non-thermal. Hence, it could be used in principle to extract work from nothingness by means of a suitable thermal engine. 
We will call such a hypothetical device a \emph{collapse engine}.

To connect this to our previous setting, notice that the evolution equation (\ref{GRWapprox}) defines a quantum channel
\be
\Omega_{\delta t}(\rho_t) = \rho_{t+\delta t},
\ee
with $\rho_t$ the solution of eq.~(\ref{GRWapprox}) for the initial state $\rho_0$.

We suppose that the particle under consideration is subject to a harmonic potential, i.e., $H = \frac{m \omega^2}{2} X^2 + \frac{1}{2m}P^2$, and that, despite the GRW dynamics, the bath's temperature $T$ is constant. A physical justification for this last assumption is that the temperature-increasing GRW dynamics is countered by radiation from the bath into outer space. Hence, as a function of time, the temperature will converge to a stationary value $T=T_{eq}$ above the temperature of the cosmic microwave background (CMB)~\footnote{If we model the bath as a grey body, then it radiates energy at a rate of $\sigma (T^4-T_c^4)$, where $T_c$ is the temperature of the CMB and $\sigma$ is a constant that depends on how well isolated the bath is. The power transferred to the bath by the GRW dynamics is, on the other hand, independent of $T$ and proportional to $\Lambda$. It follows that the bath will reach a stationary temperature $T_{eq}$ whose exact value will depend on both $\sigma$ and $\Lambda$.}. Finally, we suppose that our set of resource operations remains the same: that is, in spite of the modified Schr\"{o}dinger equation (\ref{GRWapprox}), we can still switch on and off any unitary interaction that commutes with the total energy of the system. Notice that the GRW dynamics can be modelled by an open system approach, where the particle is interacting with some unknown system such that the resulting evolution is given by (\ref{GRWapprox}). From this viewpoint, we are simply assuming that we still have the capacity to interact with the system in the usual way.

In these conditions, we wish to find the maximum amount of work that a collapse engine could extract if it had access to the evolution equation (\ref{GRWapprox}) for a finite amount of time $t$. From Proposition (\ref{max_work}), this amounts to computing $\lim_{\delta t\to 0}\frac{t}{\delta t}W(\Omega_{\delta t}, H)$.

First, notice that we can (reversibly) evolve the system with the Hamiltonians $H$ or $-H$, since this corresponds to a thermal operation. This implies that we can ignore the first term in the right hand side of (\ref{GRWapprox}), and what remains is a Gaussian channel given by
%XXXXXXXX
%This can be determined by studying the effect of the evolution term on the Wigner function of a Gaussian state, which will experience the transformation:
%XXXXXXXXX
\be
d\to d,\gamma\to\gamma+\left(\begin{array}{cc}0&0\\0&\frac{\hbar^2\Lambda \delta t}{2}\end{array}\right).
\ee

It follows that the energy of any input state will increase by $\Delta U\equiv\frac{\hbar^2\Lambda}{4m} \delta t$. From Proposition~\ref{Gaussian}, we can estimate the entropy increase by just considering Gaussian states. Now, the entropy of a $1$-mode Gaussian state is an increasing function of the determinant of its covariance matrix~\cite{gaussianos}, which, by the above equation, can only increase with time. Hence, $W(\Omega_{\delta t},H)\leq\Delta U$.

On the other hand, suppose we input a squeezed state with $\gamma=\mbox{diag}(1/r,\hbar^2r)$. Then the determinants of the covariance matrices of initial and final states will be $\hbar^2$ and $\hbar^2+\frac{\hbar^2\Lambda\delta t}{2r}$, respectively.
The entropy change of the state can thus be made as small as desired by increasing the value of $r$, and so the above bound can be saturated, leading to $W(\Omega_{\delta t},H)=\Delta U$. Consequently, the maximum power at which a collapse engine could in principle operate is given by
\be
\frac{dW}{dt}=\frac{\hbar^2\Lambda}{4m}.
\ee
Using the upper range estimation $\Lambda \sim 10^{6} s^{-1} m^{-2}$, we have that a collapse engine powered by a single electron would produce $\frac{dW}{dt} \sim 10^{-32} \ watt$. Assuming total control over the electrons of a macroscopic sample, %this last result implies that 
one would need a kiloton of Hydrogen to power a $40 \ watt$ light bulb. 

\vspace{10pt}
\noindent\emph{Conclusion}
\vspace{10pt}

In this Letter we addressed the problem of determining how much work can be extracted from operational -as opposed to state- resources. We proved that the solution to this problem in the asymptotic limit is given by a single-letter formula that quantifies the amount of distillable work that a channel can, in principle, generate when supplemented with thermal operations and catalyst states. Moreover, we found how this quantity can be determined for bosonic channels, and computed it exactly for the case of the GRW dynamics, hence determining the maximum power which a hypothetical collapse engine could provide for free.

Note that we have only studied operational resources regarding their capacity to generate work. An interesting topic for future research is to extend our results and draw a map of the inter-conversion relations between different operational resources. In the case of state resources there is a unique monotonic quantity, the free energy, determining the optimal rates for state transformations~\cite{Brandao11}. In this work we have identified an operational monotone, the distillable work, but we suspect that there may be many others.

\vspace{10pt}

\noindent\emph{Acknowledgements}

We thank Ralph Silva and Noah Linden for useful discussions.

\appendix

\section{The deterministic distillable work can be superactivated} 

Consider the channel $\Omega$ that takes any state of a target two-level system with Hamiltonian $H= E_1\proj{1}$ to the state $\ket{\psi} \propto \left( \ket{0} + e^{-\beta E_1/2} \ket{1} \right)$. That is, $\Omega(\rho)=\proj{\psi}$ for all $\rho$.

Now, any protocol that pretends to extract work out via a single use of $\Omega$ can be divided in three steps:

\begin{enumerate}
\item
\label{prep}
The system is prepared in a state which comprises catalysts, thermal states and the work system (in state $\ket{0}$). That is, $\rho_0=\sigma_{cat}\otimes\tau_{th}\otimes\proj{0}_w$.
\item
We apply an energy-conserving unitary $U_1$ over the whole system.
\item
\label{channel}
We apply $\Omega$, hence replacing a subsystem's state by $\ket{\Psi}$.
\item
We apply a second energy-conserving unitary $U_2$ over the whole system.
\end{enumerate}

\noindent At the end of the protocol, the work system will have evolved to $\ket{1}$.

If, rather than implementing step \ref{channel}, we add the state $\ket{\psi}$ in the preparation stage, then it is trivial to find an energy-conserving unitary $\tilde{U}_2$ that at the last step would produce exactly the same amount of work. That is, we would have extracted work from the resource state $\ket{\psi}$. 

In~\cite{Horodeckisingleshot}, however, it is shown that no work can be extracted from a single copy of $\ket{\psi}$, even with the use of catalysts. This implies that the deterministic distillable work of a single use of $\Omega$ is zero. 

Suppose, now, that we have access to two uses of $\Omega$. Then we can prepare two copies of $\ket{\psi}$, from which a non-zero amount of work can be deterministically extracted via thermal operations \cite{Horodeckisingleshot}.

\section{Gaussian channels with finite distillable work} 

Let $\Omega$ be a Gaussian channel whose action on the displacement vector $\vec{d}$ and covariance matrix $\gamma$ of the $m$-mode input state is given by

\be
d\to Xd+z, \gamma\to X\gamma X^T+Y.
\label{gauss_chan}
\ee

\noindent Suppose that $\Omega$ acts on a system with Hamiltonian $H=\vec{R}^TG\vec{R}+\vec{h}\cdot\vec{R}$, where $\vec{R}$ is the vector of optical quadratures. Under the assumption that $\tilde{G}\equiv G-X^TGX>0$, we wonder if the difference between the free energies of the input and output states is bounded, i.e., whether $\Delta F(\rho,\Omega)<K$, for some $K<\infty$.

Call $E_0$ the energy of the input state; and $\gamma, \vec{d}$, its covariance matrix and displacement vector. Then we have that 

\be
E_0=\frac{1}{2}\tr{G\gamma}+(\vec{d}-\vec{d}_0)^TG(\vec{d}-\vec{d}_0)-\bar{E},
\ee

\noindent where $\vec{d}_0\equiv -G^{-1}\vec{h}/2$, and $\bar{E} \equiv \vec{d}_0^TG\vec{d}_0$. Hence, defining $\mu_{\min}>0$ ($\mu_{\max}>0$) to be the minimum (maximum) eigenvalue of $G$ we have that

\be
\frac{E_0+\bar{E}}{\mu_{\max}}\leq\left(\frac{1}{2}\tr{\gamma}+\|\vec{d}-\vec{d}_0\|^2\right)\leq \frac{E_0+\bar{E}}{\mu_{\min}},
\ee

\noindent and consequently

\begin{align}
&\tr{\gamma}\leq O(E_0),\|\vec{d}\|\leq O(\sqrt{E_0}),\nonumber\\
&\frac{1}{2}\tr{\gamma}+\|\vec{d}-\vec{d}_0\|^2\geq O(E_0).
\label{lower}
\end{align}

We can now bound the energy difference between the input and output states. First, note that $\Delta E\equiv E_f-E_0$ can be written as:

\be
\Delta E=-\frac{1}{2}\tr{\gamma \tilde{G}}-(\vec{d}-\vec{d}_0)^T\tilde{G}(\vec{d}-\vec{d}_0)+O(\vec{d}).
\ee

\noindent Defining $\lambda_{\min}>0$ to be the smallest eigenvalue of $\tilde{G}$, we thus arrive at

\begin{align}
\Delta E&&\leq -\lambda_{\min}\left(\frac{1}{2}\tr{\gamma}+\|\vec{d}-\vec{d}_0)\|^2\right)+O(\vec{d})\nonumber\\
&&\leq-O(E_0)+O(\sqrt{E_0})=-O(E_0).
\end{align}

Let us now bound the entropy of the input state: by the subadditivity of the von Neumann entropy, $S(\rho)$ is bounded from above by $\sum_{i=1}^mS(\rho_i)$, where $\rho_{i}$ denotes the reduced density matrix of each mode $i$. $S(\rho_i)$, in turn, is bounded by the von Neumman entropy of the Gaussian state with the same first and second moments as $\rho_i$, i.e., a Gaussian state with covariance matrix $\gamma_i$. Note that $\sum_{i=1}^m\tr{\gamma_i}=\tr{\gamma}\leq O(E_0)$, where the last inequality is due to eq. (\ref{lower}). In particular, $\tr{\gamma_i}\leq O(E_0)$ for $i=1,...,m$. 

\noindent The entropy of a 1-mode Gaussian state with covariance matrix $\tilde{\gamma}$ is given by

\be
(N+1)\log(N+1)-N\log(N)= O(\log(N)),
\ee

\noindent where $N=\sqrt{\det(\tilde{\gamma})}/\hbar\leq \tr{\tilde{\gamma}}/2\hbar$ (this last inequality reflects the fact that the geometric mean of $\tilde{\gamma}$'s two eigenvalues is smaller than their arithmetic mean). It follows that

\be
S(\rho)\leq O(\log(E_0)).
\ee

Putting all together, we have that

\begin{align}
\Delta F(\rho, \Omega) &\leq \Delta E+K_BTS(\rho)&\nonumber\\
&\leq -O(E_0)+O(\log(E_0)).&
\end{align}

The last expression cannot thus take arbitrarily large values, and so the distillable work of channel $\Omega$ is bounded.

\section{Proof of Proposition 2}
Proposition 2 follows straightforwardly from the following Lemma:

\begin{lemma}
Let $\rho$ be an arbitrary state with finite first and second moments, and let $\rho_G$ be the unique Gaussian state with the same first and second moments. Then, for any Gaussian channel $\Omega$, we have that

\be
S(\rho_G)-S(\Omega(\rho_G))\geq S(\rho)-S(\Omega(\rho)).
\label{paramecio}
\ee

\end{lemma}

\begin{proof}

Let the action of $\Omega$ over the displacement vector and covariance matrix be given by eq. (\ref{gauss_chan}). Since von Neumann entropies remain the same after a unitary transformation, without loss of generality we will assume that $\rho$'s displacement vector is null. Similarly, we will take $z=0$ in the channel description (\ref{gauss_chan}) of $\Omega$. Now, let $U_n$ be the $n$-system `Gaussification' transformation described in \cite{gaussification}. Calling $\tilde{\rho}^{(n)}=U_n\rho^{\otimes n}U_n^\dagger$, we have that

\bea
L_n &\equiv& \sum_{i=1}^nS(\tilde{\rho}_i)-S(\Omega(\tilde{\rho}_i)) \nonumber\\
&\geq& S(\tilde{\rho}^{(n)})-S(\Omega^{\otimes n}(\tilde{\rho}^{(n)}) )\noindent\\
&=& S(\rho^{\otimes n})-S(\Omega^{\otimes n}(\rho^{\otimes n}))\nonumber\\
&=&n\{S(\rho)-S(\Omega(\rho))\},
\label{inter1}
\eea

\noindent where the first inequality is due to Lemma 1 in the main text, and the equality follows from the fact that, for all states $\sigma$,

\be
U_n\Omega^{\otimes n}(\sigma)U_n^\dagger=\Omega^{\otimes n}(U_n\sigma U_n^\dagger).
\ee

\noindent This identity follows from three observations: 1) any Gaussian channel with $z=0$ is the result of applying a symplectic unitary $V_S$ over the target system and an ancillary Gaussian state $\omega$ with zero displacement vector; 2) $n$ copies of $\omega$ are invariant with respect to a Gaussification operation $U^A_n$; 3) $[V_S^{\otimes n},U_n']=0$, where $U_n'=U_n\otimes U^A_n$ represents the Gaussification of $n$ copies of the system target-ancilla.

Since the von Neumann entropy is continuous in trace norm with respect to collections of states with finite second moments, by \cite{gaussification}, we end up with

\bea
\lim_{n\to\infty}\frac{L_n}{n} &= &\lim_{n\to\infty}\frac{1}{n}\sum_{i=1}^nS(\tilde{\rho}_i)-S(\Omega(\tilde{\rho}_i))\nonumber\\
&=&S(\rho_G)-S(\Omega(\rho_G)).
\label{inter2}
\eea

\noindent The lemma hence follows from (\ref{inter1}) and (\ref{inter2}).

\end{proof}

Now, let $\rho$ be an arbitrary state with finite energy (and thus finite first and second moments), and let $\rho_G$ be the Gaussian state with the same first and second moments. Then, $\rho$ and $\rho_G$ have the same average energy, and since $\Omega$ is a Gaussian channel the same is true for $\Omega(\rho)$ and $\Omega(\rho_G)$. However, by the previous Lemma, the entropic change is bigger for $\rho_G$, and so $\Delta F(\rho_G,\Omega)\geq \Delta F(\rho,\Omega)$, proving Proposition 2.

\bibliography{references}
\end{document}